\documentclass[envcountsame]{llncs}
\usepackage[utf8]{inputenc}
\usepackage{amsmath,amssymb}

\title{On the expressive power of read-once determinants \thanks{A preliminary version of this work has appeared in 20th International Symposium on
Fundamentals of Computation Theory, Proceedings. Lecture Notes in Computer Science, Volume 9210, Springer 2015, pages 95-105}}

\author{N. R. Aravind\inst{1},
  Pushkar S. Joglekar\inst{2},
 \institute{Indian Institute of Technology, Hyderabad, India\\
  \email{aravind@iith.ac.in}\and
  Vishwakarma Institute of Technology, Pune, India\\
  \email{joglekar.pushkar@gmail.com}}
}

\date{}

\usepackage{amsfonts}
\usepackage{epsfig}
\usepackage{graphicx}
\usepackage{verbatim}
\usepackage{parskip}
\usepackage{tikz}
\usepackage{bm}
\usepackage{blkarray}

\def \hs {\hspace{3mm}}

\newenvironment{proof-of}[1]{\noindent{\bf Proof of #1}\hspace*{0.5em}}{\qed\bigskip}
\newenvironment{proof-of-lemma}[1]{\noindent{\bf Proof of Lemma #1}\hspace*{1em}}{\qed\bigskip}

\begin{document}

\maketitle 
\begin{abstract}
We introduce and study the notion of read-$k$ projections of the determinant: a polynomial $f \in \mathbb{F}[x_1, \ldots, x_n]$ is called a {\it read-$k$ projection of determinant} if $f=det(M)$, where entries of matrix $M$ are either field elements or variables such that each variable appears at most $k$ times in $M$. A monomial set $S$ is said to be expressible as read-$k$ projection of determinant if there is a read-$k$ projection of determinant $f$ such that the monomial set of $f$ is equal to $S$.
We obtain basic results relating read-$k$ determinantal projections to the well-studied notion of determinantal complexity. 
We show that for sufficiently large $n$, the $n \times n$ permanent polynomial $Perm_n$ and the elementary symmetric polynomials of degree $d$ on $n$ variables $S_n^d$ for $2 \leq d \leq n-2$ are not expressible as read-once projection of determinant, whereas $mon(Perm_n)$ and $mon(S_n^d)$ are expressible as read-once projections of determinant. We also give examples of monomial sets which are not expressible as read-once projections of determinant.
\end{abstract}
\section{Introduction}
In a seminal work \cite{V79}, Valiant introduced the notion of the determinantal complexity of multivariate polynomials and proved that any polynomial $f \in \mathbb{F}[x_1, \ldots, x_n]$ can be expressed as $f= det(M_{m \times m})$, where the entries of $M$ are affine linear forms in the variables $\{x_1, x_2, \ldots, x_n \}$. The smallest value of $m$ for which $f=det(M_{m \times m})$ holds is called the determinantal complexity of $f$ and denoted by $dc(f)$.
Let $Perm_n$ denote the permanent polynomial:
\[
Perm_n(x_{11}, \ldots, x_{nn})= \sum_{\sigma \in S_n } \prod_{i=1}^{n} x_{i, \sigma(i)}
\] 

Valiant postulated that the determinantal complexity of $Perm_n$ is not polynomially bounded - i.e. $dc(Perm_n)=n^{\omega(1)}$. This is one of the most important conjectures in complexity theory. So far the best known lower bound on $dc(Perm_n)$ is $\frac{n^2}{2}$, known from \cite{MR04}, \cite{CCL08}.

Another related notion considered in \cite{V79} is projections of polynomials: A polynomial $f \in \mathbb{F} [x_1, \ldots, x_n]$ is said to be a projection of $g \in \mathbb{F}[y_1, \ldots, y_m],  m \geq n$ if $f$ is obtained from $g$ by substituting each variable $y_i$ by some variable in $\{x_1, x_2, \ldots, x_n\}$ or by an element of field $\mathbb{F}$. Valiant's postulate implies that  if $Perm_n$ is projection of the Determinant polynomial $Det_m$ then $m$ is $n^{\omega(1)}$. 
We refer to the expository article by von-zur Gathen on Valiant's result \cite{G87}.

We define the notion of read-$k$ projection of determinant, which is a natural restriction of the notion of projection of determinant. 
Let $X=\{x_1,\ldots,x_n\}$ be a set of variables and let $\mathbb{F}$ be a field.

\begin{definition}
We say that a matrix $M_{m \times m}$ is a \emph{read-$k$ matrix} over $X \cup \mathbb{F}$ if the entries of $M$ are from $X \cup \mathbb{F}$ and for every $x \in X$, there are at most $k$ pairs of indices $(i,j)$ such that $M_{i,j}=x$.
We say that a polynomial $f\in \mathbb{F}[X]$ is \emph{read-$k$ projection of $Det_m$} if there exists a read-$k$ matrix $M_{m \times m}$ over $X$ such that $f=det(M)$.
\end{definition}

{\it Remark:} We use the phrase \emph{a polynomial is expressible as read-once determinant} in place of ``a polynomial is read-$1$ projection of determinant'' in some places. Note that only a multilinear polynomial can be expressible as a read-once determinant.

The following upper bound on determinantal complexity, proved in Section 2, is one of the motivations for studying this model.
 
\begin{theorem}\label{dc-ROD} 
Let $f \in \mathbb{F}[x_1,\ldots,x_n]$. If $f$ is a read-$k$ projection of determinant, then $dc(f) \leq nk$.
\end{theorem}

The above theorem immediately shows that read-$k$ projections of determinant are not universal for any constant $k$; indeed in the case of finite fields, by simple counting arguments, we can show that most polynomials are not read-$k$ expressible for $k=2^{o(n)}$.

Ryser's formula for the permanent expresses the permanent polynomial $Perm_n$ as a read-$2^{n-1}$ projection of determinant.
In contrast, it follows from Theorem \ref{dc-ROD} that Valiant's hypothesis implies the following: $Perm_n \neq det(M_{m \times m})$ for a read-$n^{O(1)}$ matrix $M$ of \emph{any size}.
So the expressibility question is more relevant in the context of read-$k$ determinant model rather than the size lower bound question. In this paper, we obtain the following results for the simplest case $k=1$.

\begin{theorem}\label{Perm-not-ROD}
 For $n > 5$, the $n \times n$ permanent polynomial $Perm_n$ is not expressible as a read-once determinant over the field of reals and over finite fields in which $-3$ is a quadratic non-residue.
\end{theorem}

We prove Theorem \ref{Perm-not-ROD} in Section 3 as a consequence of non-expressibility of elementary symmetric polynomials as read-once determinants.


Our interest in this model also stems from the following reason.
Most of the existing lower-bound techniques for various models, including monotone circuits, depth-3 circuits, non-commutative ABPs etc, are not sensitive to the coefficients of the monomials of the polynomial for which the lower-bound is proved. 
For example, the monotone circuit lower-bound for permanent polynomial by Jerrum and Snir \cite{JS82}, carries over to any polynomial with same monomial set as permanent. The same applies to Nisan's rank argument \cite{N91} or the various lower bound results based on the partial derivative techniques (see e.g. \cite{NW97}, \cite{GR00}). 

On the other hand, for proving lower bounds on the determinantal complexity of the permanent, one must use some properties of the permanent polynomial which are not shared by the determinant polynomial. A natural question is whether there are models more restrictive than determinantal complexity (so that proving lower bounds may be easier) and which are coefficient-sensitive. Read-$k$ determinants appear to be a good choice for such a model.

In light of the above discussion and to formally distinguish the complexity of a polynomial and that of its monomial set, we have the following definition.

\begin{definition}
For $f \in F[X]$, we denote by \emph{mon(f)} the set of all monomials with non-zero coefficient in $f$.
We say that a set $S$ of monomials is \emph{expressible as read-$k$ determinant} if there exists a polynomial $f \in \mathbb{F}[X]$ 
such that 
$f$ is a read-$k$ projection of determinant and $S=mon(f)$.
\end{definition}

Let $S_n^d$ denote the elementary symmetric polynomial of degree $d$:
\[
S_n^d(x_1,x_2,\ldots, x_n) = \sum_{A\subseteq \{1,2, \ldots, n\}, |A|=d}~~ \prod_{i \in A} x_i 
\]

In Section 3, we prove the non-expressibility of elementary symmetric polynomials as read-once determinants; a contrasting result also proved in the same section is the following:

\begin{theorem}\label{monomial-existence} For all $n \geq d \geq 1$ and $|\mathbb{F}| \geq n$, the monomial set of $S_n^d$ is expressible as projection of read-once determinant.
\end{theorem}

The organization of the paper is as follows.
In Section 2, we prove Theorem \ref{dc-ROD} and make several basic observations about read-once determinants. In Section 3, as our main result, we show the non-expressibility of the elementary symmetric polynomials as read-once determinants and as a consequence deduce non-expressibility of the permanent (Theorem \ref{Perm-not-ROD}). We also prove that the monomial set of any elementary symmetric polynomial is expressible as a read-once determinant (Theorem \ref{monomial-existence}). In Section 4, we give examples of monomial sets which are not expressible as read-once determinants.

\section{Basic observations}
First we note that read-once determinants are strictly more expressive than occurrence-one algebraic branching programs which in turn are strictly more expressive than read-once formulas. By \emph{occurrence-one ABP} we mean an algebraic branching program in which each variable is allowed to repeat at most once \cite{JQS10}.(We are using the term occurrence-one ABPs rather than read-once ABPs to avoid confusion, as the latter term is sometimes used in the literature to mean an ABP in which any variable can appear at most once on any source to sink path in the ABP).  

In the following simple lemma, we compare read-once determinants with read-once formulas and occurrence-one ABPs.

\begin{lemma}\label{ROF-ABP-ROD}
Any polynomial computed by a read-once formula can be computed by an occurrence-one ABP, and any polynomial computed by an occurrence-one ABP can be computed by a read-once determinant. Moreover there is a polynomial which can be computed by read-once determinant but can't be computed by occurrence-one ABPs.
\end{lemma}  
\begin{proof}
Let $f\in \mathbb{F}[X]$ be a polynomial computed by a read-once formula or an occurrence-one ABP. Using Valiant's construction \cite{V79}, we can find a matrix $M$ whose entries are in $X \cup \mathbb{F}$ such that $f=det(M)$. We observe that if we start with a \emph{read-once} formula or an \emph{occurrence-one} ABP 
then for the matrix $M$ obtained using Valiant's construction, every variable repeats at most once in $M$. This proves that $f$ can be computed by read-once determinants. 
To see the other part, consider the elementary symmetric polynomial of degree two over $\{x_1, x_2, x_3\}$:$S_3^2(x_1,x_2,x_3)= x_1x_2 + x_1 x_3 + x_2 x_3$. It is proved in (\cite{JQS10} Appendix-B) that $S_3^2$ cannot be computed by an occurrence-one ABP. From the discussion in the beginning of Section \ref{sym-poly-section} it follows that $S_3^2$ can be computed by a read-once determinant.
%
\end{proof}

Let $X=\{x_1, x_2, \ldots, x_n\}$. Let $S = \{i_1, i_2, \ldots, i_k \} \subseteq [n]$, let $X_S$ denote the set of variables $\{x_{i_1}, x_{i_2}, \ldots, x_{i_k} \}$. We define $\frac{\partial f}{\partial X_S}$ a partial derivative of $f$ with respect to $X_S$ as 
 $ \frac{\partial f}{\partial X_S}= \frac{\partial ^k f}{\partial x_{i_1}\partial x_{i_2}\ldots \partial x_{i_k}}$. For a vector $a= (a_1, \ldots, a_k)$ with $a_i \in \mathbb{F}$ let $f |_{S=a}$ denote polynomial $g$ over variables $X \setminus X_S$  which is obtained from $f$ by substituting variable $x_{i_j}= a_j$.  

We define the set ROD to be the set of all polynomials in $\mathbb{F}[X]$ that are expressible as read-once projection of determinant. The following simple proposition shows that the set ROD has nice closure properties.
\begin{proposition}\label{closure-ROD}
Let $f$ be a polynomial over $X$ such that $f \in ROD$ and $S \subseteq X$, $|S|=k$. Let $a \in \mathbb{F}^k$. Then  $f|_{S=a}, \frac{\partial f}{\partial X_S} \in ROD$. For any polynomial $g \in ROD$ such that $fg$ is a multilinear polynomial, we have $fg \in ROD$.
\end{proposition}
\begin{proof}
Closure under substitution follows easily from the definition of ROD. Let $f= det(M)$ for a read once matrix $M$. W.l.g. assume that all the variables in $X_S$ appear in different rows and columns in $M$ (otherwise $\frac{\partial f}{\partial x_S} =0$). $(2)$ follows by noting that $\frac{\partial f}{\partial x_S}$ is a minor of $M$ obtained by removing rows and columns corresponding to variables in $X_S$. Since $fg$ is multilinear, both $f$ and $g$ must be multilinear and $Var(f)\cap Var(g)=\phi$ where $Var(f), Var(g)$ denote variable sets for $f$ and $g$ respectively. Let $f=det(M_1)$ and $g=det(M_2)$ for read once matrices $M_1$ and $M_2$. Let $M$ be a matrix obtained by putting copy of $M_1$ and $M_2$ on diagonal, so that $det(M)= det(M_1)det(M_2)$. $M$ is read once matrix since variable sets of $f$ and $g$ are disjoint.   \end{proof}

Now we observe that, if polynomial $f$ is expressible as read-$k$ projection of determinant then determinantal complexity of $f$ is upper bounded by $nk$. We use the notation $M \sim N$ to mean that $det(M)=det(N)$.

\begin{proof-of}{Theorem \ref{dc-ROD}}
Let $f=det(M)$ for a read-$k$ matrix $M$. Without loss of generality, we can assume that for some $m \leq kn$, the principal $m$ by $m$ submatrix of $M$ contains all the variables. 
Let $Q$ denote this submatrix and let $R$ denote the submatrix formed by the remaining columns of the first $m$ rows. Suppose that the number of remaining rows of $M$ is equal to $p$.
Let $T$ denote the submatrix formed by these rows. We note that $M$ has full rank, and hence the row-rank and column-rank of $T$ are both equal to $p$.

Consider a set of $p$ linearly independent columns in $T$ and let $T_1$ be the submatrix formed by these columns; let $T_2$ denote the remaining $m$ columns of $T$. 
Further, let $Q_1$ and $R_1$ denote the columns of $Q$ and $R$ respectively, corresponding to the columns of $T_1$ and similarly, let $Q_2$ and $R_2$ denote the columns of $Q$ and $R$, corresponding to the columns of $T_2$. In other words, the columns of $M$ can be permuted to obtain $M' \sim M$:

\[ M'=\left(
\begin{matrix}
Q_1|R_1 & Q_2|R_2  \\
T_1 & T_2 
\end{matrix}
\right). \]

Let $g$ denote the unique linear transformation such that $[T_2+g(T_1)]=[0]$.

Applying $g$ to the last $m$ columns of $M'$, we obtain a matrix $N \sim M'$ such that

\[N=\left(
\begin{matrix}
Q_1|R_1 &  Q_2|R_2+g(Q_1|R_1)  \\
T_1 & 0 
\end{matrix}
\right).\]

Let $det(T_1)=c \in \mathbb{F}$; clearly $c \neq 0$. 
Let $N'$ be a matrix obtained by multiplying some row of $[Q_2|R_2]+g([Q_1|R_1])$ by $c$.
The entries of $N'$ are affine linear forms, $det(N')=f$ and the dimension of $N'$ is $m \leq kn$.
This proves Theorem \ref{dc-ROD}.
\end{proof-of}

In the next lemma we show that if $f \in\mathbb{F}[x_1, \ldots, x_n]$ such that $f=det(M_{m \times m})$ for a read-once matrix $M$, then we can without loss of generality assume that $m \leq 3n$. 
\begin{lemma}\label{ROD-bound}
Let $f  \in \mathbb{F}[x_1,\ldots,x_n]$ be expressible as read-once determinant. 
Then there is a read-once matrix $M$ of size at most $3n$ such that $f=det(M)$.
\end{lemma}
\begin{proof}
The proof is on similar lines as that of Theorem \ref{dc-ROD}. We begin with the observation that we may replace $T_1$ in the matrix $M$ with any $k$ by $k$ matrix whose determinant is equal to $c$ and set $T_2=-f(T_1)$. For simplicity, we choose $T_1$ to be the diagonal matrix whose first entry is $c$ and all other entries are equal to $1$.

Consider the matrix $R_1$ - its rank is at most $m \leq n$; let $S_1$ denote the submatrix of $R_1$ formed by a maximal independent set of columns.
The key observation is that we can find a linear transformation $g$ such $f([Q_1,R_1])=g([Q_1,S_1])$.
Let $T_3$ denote the columns of $T_1$ corresponding to the columns of $[Q_1|S_1]$.

Consider the read-once matrix 

\[N_1=\left(
\begin{matrix}
Q_1|S_1 &  Q_2|R_2  \\
T_3 & -g(T_3)
\end{matrix}
\right).\]

From the previous observations, it is clear that $N \sim N_1$; the number of columns of $Q$ is $m$ and the number of columns of $S_1$ and $R_2$ are at most $m$ each; thus the dimension of $N_1$ is at most $3m \leq 3n$.
\end{proof}

\section{Elementary symmetric polynomials and permanent}
In this section we will prove our main result: the elementary symmetric polynomials $S_n^d$ for $2\leq n \leq n-2$ and the permanent $Perm_n$ are not expressible as read-once determinants for sufficiently large $n$. We will first prove that $S_4^2 \not\in ROD$ and use it to prove non-expressibility of $Perm_n$ and $S_n^d$.





We begin with following simple observation based on the closure properties of ROD.


\begin{lemma}\label{s42-imply-general}
\begin{enumerate}
\item If $S_m^k \not\in ROD$ then $S_{n}^{d}\not\in ROD$ for $d\geq k$ and $n\geq m+d-k$.
\item If $Perm_m \not\in ROD$ then $Perm_n\not\in ROD$ for $n\geq m$.  
\end{enumerate}
\end{lemma}
\begin{proof}
Let $X=\{x_1, \ldots, x_n\}$ and $A = \{i_1, i_2, \ldots, i_k \} \subseteq [n]$. It is easy to see that $\frac{\partial S_n^d}{\partial X_A}$ is the elementary symmetric polynomial of degree $d-k$ over the set of variables $X \setminus X_A$.  
Let $S_n^d \in ROD$ for some $d\geq k$ and $n \geq m+d-k$. By Proposition~\ref{closure-ROD} we know that read-once determinants are closed under partial derivatives. So we get $S_{n-(d-k)}^{d-(d-k)} = S_{n-d+k}^k \in ROD$. Let $n-d+k = m+l$ for $l \geq 0$. Now we substitute any $l$ variables to zero in $S_{n-d+k}^{k}$ to obtain the polynomial $S_m^k$. We have $S_m^k \in ROD$ as from Proposition~\ref{closure-ROD} RODs are closed under substitution. 

The proof for the second part is similar and follows easily by noting that the partial derivative of $Perm_n$ with respect to any variable $x_{i,j}$, $1\leq i,j \leq n$ is the $(n-1)\times(n-1)$ Permanent polynomial on $(n-1)^2$ variables.  
\end{proof}

From the above lemma, it is clear that, if the polynomials $Perm_n$ or $S_n^d$ are expressible as read-once determinants for some $n$ then these polynomials will be expressible as read-once determinants for some constant value of $n= O(1)$.

\subsection{Elementary symmetric polynomials} \label{sym-poly-section}
For $d=1$ or $n$, the elementary symmetric polynomial $S_n^d$ can be computed by a $O(n)$ size read-once formula so by lemma~\ref{ROF-ABP-ROD} we can express $S_n^d$ as a read-once determinant. We observe that $S_n^{n-1} \in ROD$ over any field as $S_n^{n-1} = det  \left(\begin{matrix}
D &  A \\
C &  0
\end{matrix}
\right)$ Where $D$ is $n \times n$ diagonal matrix with $(i,i)^{th}$ entry $x_i$ for $i=1$ to $n$. $C$ and $A$ are $1 \times n$ and $n \times 1$ matrices such that all entries of $C$ are $1$ and all entries of $A$ are $-1$. 
So $S_n^d \in ROD$ for $d=1,n-1, n$. In this section, we show that $S_n^d \notin ROD$ for every other choice of $d$ in the case of field of reals or finite fields in which $-3$ is quadratic non-residue. 
%
%




First we consider the case of field of real numbers. Let $S_4^2(x_1,x_2,x_3,x_4) = c'\cdot det(M)$ for a read-once matrix $M$ and a non-zero $c' \in \mathbb{R}$. Rearranging rows and columns of $M$ or taking out a scalar common from either row or column of $M$ will change the value of $det(M)$ only by a scalar, so pertaining to the expressibility question, we can do these operations freely (we will get different scalar than $c$ as a multiplier but that is not a problem).
For any $i,j \in \{1,2,3,4\}, i \neq j$, $x_ix_j \in mon(S_4^2)$. So clearly $\frac{\partial S_4^2}{\partial x_i \partial x_j} \neq 0$ which implies that the determinant of the minor obtained by removing the rows and the columns corresponding to the variables $x_i, x_j$ is non-zero. So for any $i,j \in \{1,2,3,4\}, i\neq j$, $x_i$ and $x_j$ appear in different rows and columns in $M$.   
 By suitably permuting rows and columns of $M$ we can assume that $S_4^2 = c \cdot det(N)$ for a non zero real $c$ and read-once matrix $N$ such that $(i,i)^{th}$ entry of $N$ is variable $x_i$ for $i=1$ to $4$. So $N= \left(
\begin{matrix}
x_1 & -   & -   &  -   &\beta_1\\
-   & x_2 & -   &  -   &\beta_2\\
-   & -   & x_3 &  -  &\beta_3\\
-   & -   & -   & x_4 &\beta_4\\
\alpha_1   & \alpha_2   &\alpha_3   & \alpha_4  & L
\end{matrix}
\right)
$ ~Here $L$ is a $m-4 \times m-4$ matrix, and $\alpha_i, \beta_i$ are column and row vectors of size $m-4$ for $i=1$ to $4$ and $-$ represents arbitrary scalar entry. Let $p=m-4$. Index the columns and the rows of the matrix using numbers $1, 2, \ldots, m$.  
 Let $S$ denote a set of column and row indices corresponding to submatrix $L$. For any set $\{a_1,a_2, \ldots, a_k\} \subset \{1,2,3, 4\}$, let $N_{a_1,a_2, \ldots, a_k}$ denote a minor of $N$ obtained by removing rows and columns corresponding to indices $\{1,2,3,4\} \setminus \{a_1, \ldots, a_k \}$ from $N$. 

\begin{definition}
Let $X=\{x_1, \ldots, x_n \}$ and $M$ be a matrix with entries from $X \cup \mathbb{F}$. For $a = (a_1, a_2, \ldots, a_n) \in \mathbb{F}^n$ let $M_a$ be the matrix obtained from $M$ by substituting $x_i= a_i$ for $i=1$ to $n$. Let $maxrank(M)$(respect. $minrank(M)$) denote the maximum(respect. minimum) rank of matrix $M_a$ for $a \in \mathbb{F}^n$.
\end{definition}

Now we make some observations regarding ranks of various minors of $N$. 
\begin{lemma}\label{minor-ranks}
For $i,j \in \{1,2,3,4\}$, $i\neq j$ we have
\begin{enumerate}
\item $maxrank(N_{i,j})$=$minrank(N_{i,j})$=$p+2$
\item $maxrank(N_{i})$=$minrank(N_{i})$=$p$
\item $rank(L) \in \{p-1, p-2\}$
\end{enumerate}
\end{lemma}
\begin{proof}
Let $\{k,l\}= \{1,2,3,4\}\setminus \{i,j\}$. Monomial $x_kx_l \in mon(S_4^2)$, so the matrix obtained from $N_{i,j}$ by any scalar substitution for $x_i$ and $x_j$ has full rank. So we have $minrank(N_{i,j}) =p+2$. 
Since $N_{i,j}$ is a $(p+2) \times (p+2)$ matrix with minrank $p+2$, clearly $maxrank(N_{i,j})=minrank(N_{i,j})=p+2$.
To prove the second part, note that the matrix $N_i$ can be obtained by removing a row and a column from the matrix $N_{i,j}$. So clearly $minrank(N_i)\geq minrank(N_{i,j})-2 = p$. As $S_4^2$ doesn't contain any degree $3$ monomial we have $maxrank(N_i) \leq p$. Hence $minrank(N_i)=maxrank(N_i)=p$.

The matrix $N_i$ can be obtained from $L$ by adding a row and a column so $rank(L)\geq minrank(N_i)-2 = p-2$. 
Since monomial $x_1x_2x_3x_4 \not\in mon(S_4^2)$, $L$ can not be full-rank matrix so $rank(L) \leq p-1$. Thus proving the lemma. 
\end{proof}

Suppose $rank(L)=p-1$. By $cspan(L), rspan(L)$ we denote the space spanned by the columns and the rows of $L$ respectively. 
Next we argue that for any $i \in \{1,2,3,4 \}$, $\alpha_i  \in cspan(L)$ iff $\beta_i \not \in rspan(L)$. 
To show that we need to rule out following two possibilities 
\begin{enumerate}
\item $\alpha_i \not \in cspan(L)$ and $\beta_i \not \in rspan(L)$. In this case clearly $minrank(N_i)= rank(L) +2= p+1$, a contradiction since $minrank(N_i)=p$ by lemma \ref{minor-ranks}.
\item $\alpha_i \in cspan(L)$ and $\beta_i \in rspan(L)$. As $\beta_i \in rspan(L)$ we can use a suitable scalar value for $x_i$ so that vector $[x_i ~~\beta_i]$ is in the row span of the matrix $[\alpha_i ~L]$. Moreover $rank([\alpha_i ~ L]) = p-1$ as $\alpha_i \in cspan(L)$. So we have 
$minrank(N_i)= rank([\alpha_i ~L]) = rank(L) =p-1$. But we know that $minrank(N_i)$ is $p$. 
\end{enumerate}

So we have $\alpha_i  \in cspan(L)$ iff $\beta_i \not \in rspan(L)$. From this it follows immediately that either there exist at least two $\alpha_i$'s $\in cspan(L)$ or there exists atleast two $\beta_i$'s $\in rspan(L)$. So w.l.o.g. assume that for $i\neq j$, $\alpha_i, \alpha_j \in cspan(L)$, So $rank[\alpha_i ~\alpha_j ~ L]= rank(L)=p-1$. Matrix $N_{i,j}$ can be obtained from $[\alpha_i ~\alpha_j ~ L]$ by adding two new rows, so $maxrank(N_{i,j})\leq rank([\alpha_i ~\alpha_j ~ L])+2 = p+1$, a contradiction.
This proves that $rank(L)$ can not be $p-1$. 

Now we consider the other case. Let $rank(L)=p-2$. By applying row and column operation on $N$ we can reduce block $L$ to a diagonal matrix $D$ with all non zero entries $1$. Further applying row and column transformations we can drive entries in the vectors $\alpha_i$'s and $\beta_i$'s corresponding to nonzero part of $D$ to zero. Note that now we can remove all non-zero rows and columns of matrix $D$ still keeping the determinant same. As a result we have $S_4^2 = c_1 \cdot det(N')$ where $N'$ has following structure

 
 \[
\begin{blockarray}{ccccc cc}
 & &  &   &  &  &  \\ 
 \begin{block}{c(c c c c | c c)}
     	&	x_{1}+a_1	& 			& 			&  			&  	 \beta_{1,1} & \beta_{1,2}   \\ 
    	&			& x_2+a_2		& 			& 			&     \beta_{2,1} & \beta_{2,2} \\ 
   	 	& 			& 			& 	x_3+a_3		& 			&    \beta_{3,1} & \beta_{3,2}	 \\ 
    	& 			&			&		 	 &x_4+ a_4 			  &    \beta_{4,1} &\beta_{4,2}\\\cline{1-7}
  & 	\alpha_ {1,1}& \alpha_{2,1}	&	\alpha_{3,1} & \alpha_{4,1} & 0 & 0\\ 
  & 	\alpha_ {1,2}& \alpha_{2,2}	&	\alpha_{3,2} & \alpha_{4,2} & 0 & 0\\  
 \end{block}
\end{blockarray}
\]

Note that the coefficient of the monomial $x_ix_j$ in $N'$ is the determinant of the minor obtained by removing rows and columns corresponding to $x_i$ and $x_j$ from $N'$. It is equal to $(\alpha_{k,1}\alpha_{l,2}-\alpha_{k,2}\alpha_{l,1}).(\beta_{k,1}\beta_{l,2}-\beta_{k,2}\beta_{l,1})$ where $\{k,l\}= \{1,2,3,4\}\setminus \{i,j\}$. 
 It is easy to see that in fact without loss of generality we can assume that $\alpha_{1,1}= \beta_{1,1}=\alpha_{2,2}=\beta_{2,2}=1$ and $\alpha_{1,2}= \beta_{1,2}=\alpha_{2,1}=\beta_{2,1}=0$ (again by doing column and row transformations). So finally we have $S_4^2= c\cdot det(N)$ where $c$ is a non zero scalar, $N$ is a matrix as shown below, and $a_1, \ldots, a_4$ are real numbers.

 \[S_4^2(x_1,x_2,x_3,x_4)=c\cdot det~ 
\begin{blockarray}{ccccccc}
 & &  &   &  &  &  \\ 
 \begin{block}{c(cccc|cc)}
     	&	x_{1}+a_1	& 			& 			&  			&  	 1 & 0   \\ 
    	&			& x_2+a_2		& 			& 			&    0 & 1 \\ 
   	 	& 			& 			& 	x_3+a_3		& 			&    p' & r'	 \\ 
    	& 			&			&		 	 &x_4+ a_4 			  &    q' &s'\\\cline{1-7} 
  & 	1 & 0	&	p & q & 0 & 0\\
  & 	0& 1	&	r & s & 0 & 0\\  
 \end{block}
\end{blockarray}
\]
 
Comparing coefficients of monomials $x_ix_j$ for $i \neq j$ in $S_4^2$ and the determinant of corresponding minors of matrix $N$, we get following system of equations
$c=1, p.p'=q.q'=r.r'=s.s'=1$ and $(ps-rq)(p's'-r'q')=1$. Substituting $p'=1/p, q'=1/q$ etc in the equation above, we have $(ps-rq)(1/ps- 1/rq) = 1$ which imply $(ps - rq)^2 = -(ps)(rq)$ i.e. $(ps)^2 + (rq)^2= (ps)(rq)$ which is clearly false for non-zero real numbers $p,q,r,s$ (as $(ps)^2 + (rq)^2 \geq 2(ps)(rq)$).(Note that we need $p,q,r,s$ to be non zero since we have $pp'=qq'=rr'=ss'=1$.) This proves that $S_4^2 \not\in ROD$ over $\mathbb{R}$. 
  
In the case of finite fields $\mathbb{F}$ in which $-3$ is a quadratic non-residue the argument is as follows. We have the equation $(ps - rq)^2 = -(ps)(rq)$ as above. Let $x=ps$ and $y=rq$, so we have $y^2-xy+x^2=0$. Considering this as a quadratic equation in variable $y$, the equation has a solution in the concerned field iff the discriminant $\Delta = -3x^2$ is a perfect square, that happens only when $-3$ is a quadratic residue. So if $-3$ is a quadratic non residue, the above equation doesn't have a solution, leading to a contradiction. So we have the following theorem.

\begin{theorem}\label{s42-not-ROD}
The polynomial $S_4^2$ is not expressible as a read-once determinant over the field of reals and  over finite fields in which $-3$ is a quadratic non-residue.
\end{theorem}

We note that we can express $S_4^2$ as a read-once determinant over $\mathbb{C}$ or e.g. over $\mathbb{F}_3$ by solving the quadratic equation in the proof of Theorem \ref{s42-not-ROD}.
 \[S_4^2(x_1,x_2,x_3,x_4)=det~ 
\begin{blockarray}{ccccccc}
 & &  &   &  &  &  \\ 
 \begin{block}{c(cccccc)}
     	&	x_{1}	& 	0		& 	0		&  	0		&  	 1 & 0   \\ 
    	&	0		& x_2		& 	0		& 	0		&    0 & 1 \\ 
   	 	& 	0		& 	0		& 	x_3		& 	0		&    1 & r^{-1} \\ 
    	& 	0		&	0		&	0	 	 &x_4      &    1 &1\\ 
  & 	1 & 0	&	1 & 1 & 0 & 0\\
  & 	0& 1	&	r & 1 & 0 & 0\\  
 \end{block}
\end{blockarray}
\]
In the case $\mathbb{F}= \mathbb{C}$ choose $r=\frac{1+\sqrt{3}i}{2}$ and in case of $\mathbb{F}_3$ choose $r= 2$ (mod $3$).  

\begin{remark}
We speculate that it should be possible to prove $S_6^2 \not \in ROD$ over \emph{any field} using similar technique as in the proof of Theorem \ref{s42-not-ROD} and that would immediately give us (slightly weaker) non-expressibility results for the general elementary symmetric polynomials and the permanent polynomial as compared to the Theorems \ref{sym-final-theorem}, \ref{Perm-not-ROD}. But we haven't worked out the details in the current work.
\end{remark}
 
Theorem \ref{s42-not-ROD} together with Lemma \ref{s42-imply-general} proves the desired non-expressibility result for elementary symmetric polynomials. 
\begin{theorem}\label{sym-final-theorem}
The polynomial $S_n^d \in \mathbb{F}[x_1,x_2,\ldots, x_n]$ is not expressible as a read-once determinant for $n\geq 4$ and $2\leq d \leq n-2$ when the field $\mathbb{F}$ is either the field of real numbers or a finite field in which $-3$ is a quadratic non-residue.  
\end{theorem} 
In contrast, we show that the monomial set of $S_n^d$ is expressible as a read-once determinant.

\begin{proof-of}{Theorem \ref{monomial-existence}}
Let $k=n-d$ and $t=k+1$.

Let $D$ be a $n \times n$ diagonal matrix with $(i,i)^{th}$ entry $x_i$ for $i=1$ to $n$.
and $M=\left(\begin{matrix}
D &  A \\

C &  B
\end{matrix}
\right)$, where $A$, $B$, $C$ are constant block matrices of dimensions $n \times t$, $t \times t$ and $t \times n$, respectively. We shall choose $A,B,C$ such that $mon(det(M))=mon(S_n^d)$. Let $B=J_{t \times t}$ be the matrix with all $1$ entries. Let $C$ be such that $rank(C)=k$, $rank(C B)=t$ and such that any $k$ vectors in $Col(C)$ are linearly independent, where $Col(C)$ denotes set of column vectors of $C$. 
For example, we can let the $i^{th}$ column vector of $C$ be $(1,a_i,a_i^2,\ldots,a_i^{k-1},0)^T$ for distinct values of $a_i$.
Finally, let $A=C^T$.

It is clear that $det(M)$ is symmetric in the $x_i$'s; thus it suffices to prove that $x_1 x_2 \ldots x_i$ is a monomial of $det(M)$ if and only if $i=d$. For $1 \leq i \leq n$, consider the submatrix $M_i$ of $M$ obtained by removing the first $i$ rows and first $i$ columns.
We observe that $det(M_i)=x_{i+1} det(M_{i+1})$.
Let $r$ denote the minimum value of $i$ such that $det(M_i) \neq 0$. Then it can be seen that $x_1x_2 \ldots x_i$ is a monomial of $det(M)$ if and only if $i=r$.

We now prove that $det(M_i)=0$ if and only if $i>d$. Let $N_i$ denote the matrix formed by the last $t$ rows of $M_i$. Then $det(M_i)=0$ if and only if $rank(N_i)<t$. But $rank(N_i)=rank(Col(N_i))$ and by construction, $rank(Col(N_i))<k$ if and only if $n-i <k$, i.e. if $i>n-k=d$. This completes the proof of Theorem \ref{monomial-existence}.
\end{proof-of}

\subsection{Non-expressibility of Permanent as ROD}
Now we prove the non-expressibility result for $Perm_n$ (Theorem \ref{Perm-not-ROD}).

\begin{proof-of}{Theorem \ref{Perm-not-ROD}:}
We observe below that the elementary symmetric polynomial $S_4^2$ is a projection of the read-once $6 \times 6$ Permanent over reals. $
4 S_4^2(x_1,x_2,x_3,x_4)=Perm~ 
\begin{blockarray}{ccccccc}
 & &  &   &  &  &  \\ 
 \begin{block}{c(cccccc)}
     	&	x_1	& 	0		& 	0		&  	0		&  	 1 & 1   \\ 
    	&	0		& x_2		& 	0		& 	0		&    1 & 1 \\ 
   	 	& 	0		& 	0		& 	x_3		& 	0		&    1 & 1 \\ 
    	& 	0		&	0		&	0	 	 &x_4      &    1 &1\\ 
  & 	1 & 1	&	1 & 1 & 0 & 0\\
  & 	1& 1	&	1 & 1 & 0 & 0\\  
 \end{block}
\end{blockarray}
$~ So clearly if $Perm_{6}$ is a read-once projection of determinant then $S_4^2$ also is a read-once projection of determinant. But by Theorem \ref{s42-not-ROD} we know that $S_4^2 \not\in ROD$. So we get $Perm_6 \not\in ROD$. From lemma \ref{s42-imply-general} it follows that $Perm_n \not\in ROD$ for any $n>5$.
\end{proof-of}


\section{Non-expressible monomial sets}
We have seen that the elementary symmetric polynomials and the Permanent polynomial can not be expressed as read-once determinants but their monomial sets are expressible as ROD. In this section we will give examples of monomial sets which can not be expressed as read-once determinant.  
Let $f \in \mathbb{F}[x_1,\ldots,x_n]$. We say that $f$ is $k$-full if $f$ contains every monomial of degree $k$ and we say that $f$ is $k$-empty if $f$ contains no monomial of degree $k$.

\begin{theorem}\label{theorem-monomial-non-ROD}
Let $f\in \mathbb{F}[x_1, \ldots, x_n]$ and $f \in ROD$ be such that $f$ is $n$-full, $(n-1)$-empty and $(n-2)$-empty. Then $f$ can not be $k$-full for any $k$ such that $\lfloor \frac{n-1}{2} \rfloor \leq k < n$. 
\end{theorem}
\begin{proof}
Let $f=det(M_{m \times m})$ for a read-once matrix $M$. As $x_1x_2\ldots x_n \in mon(f)$, without loss of generality assume that the $(i,i)^{th}$ entry of $M$ is $x_i$ for $i=1$ to $n$.  Since minor corresponding to $x_1x_2\ldots x_n$ is invertible we can use elementary row and column operations on $M$ to get a matrix $N_{n \times n}$ such that $f=det N$ and the $(i,i)^{th}$ entry of $N$ is $a_i x_i + b_i$ for $a_i,b_i \in \mathbb{F}$ and $a_i \neq 0$. All the other entries of $N$ are scalars. The assumption that $f$ is $(n-1)$-empty implies that $b_i=0$ for $i=1$ to $n$. Since $f$ is $(n-2)$-empty, we also have $N(i,j)N(j,i)=0$ for all $i \neq j$, $1\leq i,j \leq n$.
So at least $\genfrac{(}{)}{0pt}{}{n}{2}$ entries of $N$ are zero. So there is a row of $N$ which contains at least $\lceil \frac{n-1}{2} \rceil$ zeros. Let $i$ be the index of that row. For $l \geq \lceil \frac{n-1}{2} \rceil$, let $a_1, a_2, \ldots, a_l$  be the column indices such that $N(i,a_i)=0$. Note that $i \not\in \{a_1, \ldots, a_l\}$. 
We want to prove that $f$ is not $k$-full for $\lfloor \frac{n-1}{2} \rfloor \leq k < n$. Let $S$ be a subset of $\{a_1, \ldots, a_l\}$ of size $n-k-1$. Note that we can pick such a set since $l \geq \lceil \frac{n-1}{2} \rceil$. Let $T=S \cup \{i\}$. Let $m=\prod_{j \not\in T} x_j$. Let $N'$ be the minor obtained by removing all the rows and columns in $\{1,2, \ldots, n\}\setminus T$ from $N$. Clearly $m\not\in mon(f)$ iff the constant term in the determinant of $N'$ is zero. Note that $N'$ contains a row with one entry $x_i$ and the remaining entries in the row are zero. So clearly the constant term in the determinant of $N'$ is zero. This shows that degree $k$ monomial $m\not\in mon(f)$. This proves the Theorem.
\end{proof}
Let $f= x_1+ x_2+ x_3 + x_4 + x_1x_2x_3x_4$. $f$ is $4$-full, $3$-empty, $2$-empty and $1$-full.   
Applying above theorem for $n=4$, we deduce that the set $mon(f)=\{x_1x_2x_3x_4,x_1,x_2,x_3,x_4\}$ is not expressible as a read-once determinant.

\section{Discussion and Open Problems}
 Under Valiant's hypothesis we know that $Perm_n$ cannot be expressed as a read-$n^{O(1)}$ determinant. Proving non-expressibility of $Perm_n$ as a read-$k$ determinant for $k>1$ unconditionally, is an interesting problem. In fact even the simplest case $k=2$ might be challenging. The corresponding PIT question of checking whether the determinant of a read-$2$ matrix is identically zero or not, is also open \cite{G00}.  

For the elementary symmetric polynomial of degree $d$ on $n$ variables, Shpilka and Wigderson gave an $O(n d^3 \log d)$ arithmetic formula \cite{SW01}. Using universality of determinant, we get an $O(n d^3 \log d)$ upper bound on $dc(S_n^d)$, in fact for non constant $d$ this is the best known upper bound on $dc(S_n^d)$ as noted in \cite{J09}. Answering the following question in either direction is interesting: Is $S_n^d$ expressible as read-$k$ determinant for $k>1$? If the answer is NO, it is a nontrivial non-expressibility result and if the answer is YES, for say $k=O(n^2)$, it gives an $O(n^3)$ upper bound on $dc(S_n^d)$, which is asymptotically better than $O(n d^3 \log d)$ for $d=\frac{n}{2}$.   

Another possible generalization of read-once determinants is the following. Let $X= \{x_{i,j}| 1 \leq i,j \leq n\}$ and consider the matrix $M_{m \times m}$ whose entries are affine linear forms over $X$ such that the coefficient matrix induced by each variable has rank one. That is if we express $M$ as $B_0 + \sum_{1\leq i,j \leq n} x_{i,j} B_{i,j}$ then $rank(B_{i,j})=1$ for $1\leq i,j \leq n$. $B_0$ can have arbitrary rank. The question we ask is: can we express $Perm_n$ as the determinant of such a matrix $M$? This model is clearly a generalization of read-once determinants
and has been considered by Ivanyos, Karpinski and Saxena \cite{IKS10}, where they give a deterministic polynomial time algorithm to test whether the determinant of such a matrix is identically zero. 
It would be interesting to address the question of expressibility of permanent in this model.

\end{document}